\documentclass[a4paper]{article}
\usepackage{amsmath,amssymb,amsthm}
\usepackage{colortbl}
\usepackage{tikz}
\usepackage{enumitem}
\usepackage{authblk}
\usepackage{algorithm}
\usepackage{algorithmic}
\usepackage{multirow}
\usepackage{bbm}
\usepackage{longtable}
\usepackage{cite}
\usepackage{algorithm}
\usepackage{algorithmic}

\DeclareMathOperator*{\argmax}{argmax}
\DeclareMathOperator*{\argmin}{argmin}

\newcommand{\real}{\mathbb{R}}
\newcommand{\nat}{\mathbb{N}}

\newcommand{\set}[2]{\left\{#1 \; \left|\;\; #2 \right.\right\}}
\newcommand{\abs}[1]{\left|#1\right|}
\newcommand{\card}[1]{\abs{#1}}

\newtheorem{theorem}{Theorem}

\newtheorem{lemma}[theorem]{Lemma}
\newtheorem{corollary}[theorem]{Corollary}
\newtheorem{proposition}[theorem]{Proposition}

\newtheorem{assumption}{Assumption}

\title{On the Complexity and Approximation of the Maximum Expected Value All-or-Nothing Subset}
\author[1]{Noam Goldberg}
\author[2]{Gabor Rudolf}
\affil[1]{Bar-Ilan University, Ramat Gan, Israel\\
 \texttt{noam.goldberg@biu.ac.il}}
\affil[2]{Ko\c c University, Istanbul, Turkey\\
\texttt{grudolf@ku.edu.tr}}
\begin{document}
\maketitle
\begin{abstract}
An unconstrained nonlinear binary optimization problem of selecting a  maximum expected value subset of items is considered. Each item is associated with a profit and probability. Each of the items succeeds or fails independently with the given probabilities, and the profit is obtained in the event that all selected items succeed. The objective is to select a subset that maximizes the total value times the product of probabilities of the chosen items. The problem is proven NP-hard by a nontrivial reduction from subset sum. Then we develop a fully polynomial time approximation scheme (FPTAS) for this problem.


\end{abstract}

\section{Introduction}

In the \emph{maximum expected value all-or-nothing subset} problem, a decision maker seeks to maximize the expected value of a subset of activities $[n]=\{1,\ldots,n\}$, where  each activity $i\in [n]$ is associated with a positive profit $c_i$ and probability of success $ p_i$. The profits are earned in an all-or-nothing fashion -- the overall success of a subset of activities depends on the individual success of all of its independent member activities. Accordingly, the problem is
\[
\max_{S\in 2^{[n]}}\sum_{i\in S}c_i\prod_{j\in S}p_j. 
\] 
The problem arises in the design of serial reliability (or 1-out-of-$n$) systems in which each component may have a different value and reliability but deriving value from the system depends on all of the selected components being operational. For example, this objective function arises in failure-aware barter exchanges such as kidney-exchange cycles in which the failure of a single pair to barter may cause the entire chain or cycle of transactions to fail~\cite{Dickerson13}. In that setting arcs of a directed graph represent possible transplants and donations; $c_i$ would be the value of transplant $i$ (that connects some donor-patient pair), and $p_i$ is the probability of the transplant $i$ taking place. The current all-subset setting corresponds to a special case of a complete directed graph of possible transplants, where a particular node (patient-donor pair) is connected to all other graph nodes through certain (probability-one) arcs. Another setting is a utility-maximizing evader who may select a subset of elicit activities that are under inspection, and the evader does not receive any value if one or more of the selected covert activities are exposed. This problem is proposed as an extension of the basic model studied in~\cite{Goldberg17}. The problem may arise in network settings and a special case is that of disjoint edges that is the subject of the current paper, which is next shown to be NP-hard.

We find it convenient to formulate 
the problem as the (unconstrained) nonlinear mathematical program with binary decision variables $x_1,\ldots,x_n$, 
\begin{align}
\label{FORM3}
& \max_{x\in\{0,1\}^n} \sum_{i=1}^n c_{i}x_{i}\prod_{j=1}^np_j^{x_j}.
\end{align}

The complexity of several related, yet different, problems has been investigated in the literature. The minimization of a (continuous) positive bilinear objective function of two variables subject to linear inequality constraints has been shown to be (strongly) NP-hard in~\cite{Matsui96}. The maximization of a product of linear functions of binary decision variables has been shown to be NP-hard in~\cite{Hammer02}. Half-product pseudo-Boolean function minimization, a special case of unconstrained quadratic binary minimization, has been shown to be NP-hard in~\cite{Badics98}. For an extensive survey of pseudo-Boolean optimization including these special cases the reader may refer to~\cite{Boros02}.

Related cost-reliability problems, with a different objective function than~\eqref{FORM3}, include variants that have been shown to be solvable in polynomial time. Let $\mathcal S\subseteq 2^{[n]}$ denote a collection of feasible item subsets.
Then a general 
cost-reliability ratio minimization problems takes the form
\begin{equation}
\label{RATIOPROB}
\min_{S\in \mathcal S}\frac{\sum_{i\in S}c_i}{\prod_{i\in S} p_i}.
\end{equation}    
 Such problems include the minimization of the spanning-tree cost to reliability ratio~\cite{Chandrasekaran84} when $\mathcal S$ is the set of all  spanning trees of a given graph. Katoh \cite{Katoh92} considers a general cost-reliability ratio minimization problem under the  assumption that given $E\subseteq [n]$, the problem of determining $S\subseteq E$ with $S\in \mathcal S$ and minimum $\sum_{i\in S}c_i$ can be solved in polynomial time in $\card E$. In \cite{Katoh92} a fully polynomial time approximation scheme (FPTAS) is developed for this problem, but the computational complexity was unresolved and appears to remain open. Also~\cite{Katoh87} develops an FPTAS for a general quasiconcave minimization problem. 
 Note that in contrast, the maximization variant of the ratio problem~\eqref{RATIOPROB} with $\mathcal S=2^{[n]}$ can be solved in polynomial time via the Dinkelbach algorithm; see~\cite{Hansen14} and references therein.

We first establish the NP-hardness of the maximum expected value all-or-nothing subset. 
Then we develop an FPTAS for this problem.

\section{Maximum Expected Value All-Or-Nothing Subset - Complexity}

 First, observe that the objective function of problem~\eqref{FORM3} can be equivalently replaced (maintaining all optimal solutions) by the concave objective
\begin{equation}
\label{LOGTRANS}
z(x)=\ln\left(\sum_{i=1}^nc_{i}x_{i}\right)+\sum_{i=1}^n\ln p_{i}x_{i}.
\end{equation}

Note that if for some $i\in [n]$, $p_i=1$, then evidently $x^*_i=1$ in every $x^*$ is optimal for~\eqref{FORM3}, while similarly $p_i=0$ implies that $x^*_i=0$. Therefore, the following assumption is without loss of generality:
\begin{assumption}
	For $i\in [n]$, the probabilities satisfy the inequalities  $0<p_i<1$.
\end{assumption} 

\medskip

For fixed $M>1$ and $y>0$ let 
\[
f(y)=\ln y-\frac{y}{M}.
\] The following lemma establishes the optimal value of $f$ and also that to determine a maximizer of $f$ over the integers it suffices to be able to approximately evaluate $f$ with precision that is bounded by a function of $M$.

\begin{lemma}\label{CNVXLEM} Let $M > 1$ be an integer. Then, the function $f$ is concave, with a unique maximum at $f(M)=\ln(M)-1$. Furthermore, for any positive integer $N\neq M$ we have $f(M)-f(N)\geq\frac{1}{5M^2}$.
\end{lemma}
\begin{proof}
	Since $f'(y)=\frac{1}{y}-\frac{1}{M}$ has a unique zero at $y=M$, and $f''(y)=\frac{-1}{y^2}<0$ holds for all $y>0$, the first part of our claim immediately follows. Keeping in mind that $f(y)$ is concave with a unique maximum at $y=M$, for any positive integer $N<M$ we have 
	$$f(M)-f(N)\geq f(M-\frac{1}{2})-f(M-1)\geq\frac{1}{2} f'(M-\frac{1}{2}) =\frac{1}{2}\left(\frac{1}{M-\frac{1}{2}}-\frac{1}{M}\right)=\frac{1}{4M^2-2M}.$$
	Similarly, for any integer $N>M$ we have
	$$f(M)-f(N)\geq f(M+\frac{1}{2})-f(M+1)\geq\frac{-1}{2} f'(M+\frac{1}{2}) =\frac{-1}{2}\left(\frac{1}{M+\frac{1}{2}}-\frac{1}{M}\right)=\frac{1}{4M^2+2M}.$$
	As $\frac{1}{4M^2-2M}\geq\frac{1}{4M^2+2M}\geq\frac{1}{5M^2}$ holds for any integer $M>1$, the proposition follows.
\end{proof}

In order to prove NP-hardness, we first show that a given instance of the subset sum problem can be decided by solving a maximum expected value all-or-nothing subset problem with the logarithmically
 transformed  objective~\eqref{LOGTRANS} and the $\ln p_i$ values as the input parameters. 

\begin{lemma}\label{MAXPROP} Let $c_1,\dots,c_n$ and $M$ be positive integers. Then, there exists an $x\in\{0,1\}^n$ such that $\sum_{i=1}^nc_ix_i=M$ if and only if the optimal objective value of the following maximization problem equals $\ln M-1$.
\begin{equation}\label{maximization}\max\limits_{{x}\in\{0,1\}^n}\ln\left(\sum\limits_ {i=1}^nx_ic_i\right)-
\frac{1}{M}\sum\limits_ {i=1}^nx_ic_i
\end{equation}
\end{lemma}
\begin{proof} First, 
by Lemma~\ref{CNVXLEM} $f$ has a unique maximum at $y=M$, $f(M)=\ln(M)-1$. Also note that the objective function in \eqref{maximization} can be written as $f\left(\sum\limits_ {i=1}^nx_ic_i\right)$.

Now assume that $\sum_{i=1}^nc_ix_i=M$ holds for some ${x}\in\real^n$.  Since ${x}$ is a solution of \eqref{maximization} with objective value $\ln(M)-1$, it is also an optimal solution according to our observation. 

Similarly, if ${x}\in\{0,1\}$ is an optimal solution of \eqref{maximization} with objective value $\ln(M)-1$ then by our observation we have $\sum\limits_{i=1}^nx_ic_i=M$.
\end{proof}

The equivalence of the optimization problems following the log transformation of the objective~\eqref{LOGTRANS} and the fact that the $c_i$ values are integer together imply the following corollary of Lemma~\ref{MAXPROP}.

\begin{corollary}
\label{COR1}
Let $c_1,\dots,c_n$ and $M$ be positive integers. There exists a subset $I\subseteq [n]$ such that $\sum_{i\in I}c_i=M$ if and only if the optimal objective value of the problem 
\begin{equation}\label{maximization2}\max\limits_{{x}\in\{0,1\}^n}\left(\sum_ {i=1}^n c_ix_i\right)\prod_ {i=1}^ne^{-\frac{ c_i}{M} x_i}
\end{equation}
is greater than $\max\{(M-1)e^{-1+1/M},(M+1)e^{-1-1/M}\}$. Equivalently,~\eqref{maximization2} has an optimal objective value of $M e^{-1}$ if and only if~\eqref{maximization} has an optimal objective \[
\max_{x\in\{0,1\}^n}z(x)>\max\left\{\ln(M-1)-\frac{M-1}{M},\ln(M+1)-\frac{M+1}{M}\right\}.
\] 
\end{corollary}

To prove that~\eqref{FORM3} is NP-hard it has to be shown that the reduction is polynomial time. However, as the input parameters $e^{-c_i/M}$ for each $i\in [n]$ cannot be exactly represented using a polynomial number of bits we employ a simple rounding argument. 

For a given $K>0$ let 
\begin{align*}
\hat p_i = \frac{\lfloor K e^{-\frac{c_i}{M}} \rfloor}{K} , && \text{and} &&
\hat z(x)=\ln(\sum_{i=1}^n c_ix_i)+\sum_{i=1}^n\ln \hat p_i x_i.
\end{align*} Observe that $p_i-\frac{1}{K} \leq \hat p_i \leq p_i$.  
The following Lemma establishes the existence of a $K$ that is polynomial in the input size for which the maximizers of $\hat z$ and $z$ coincide.

\begin{lemma}
\label{LEM1}
For $c_{\text{max}}< M$, there exists a positive $K \in O ( n  M^2)$ that satisfies for all 
$x^*$ that are optimal for~\eqref{FORM3}, 
\[
z(x^*)-\max_{x\in\{0,1\}^n}\hat z(x) < \ln M - \max\left\{\ln(M-1)+\frac{1}{M}, \ln(M+1)-\frac{1}{M}\right\}.
\]  In particular, this inequality holds for any $K > \frac{5nM^2}{1-1/(10nM^2)}$.
\end{lemma}
\begin{proof} Consider an $x\in\{0,1\}^n$. Then,
\begin{align*}
 z(x)-\hat z(x)&=\ln\left(\sum_{i=1}^nc_ix_i\right)+\sum_{i=1}^{n}\ln p_i x_i -\ln\left(\sum_{i=1}^nc_ix_i\right)-\sum_{i=1}^{n}\ln\hat p_ix_i\\
 & \leq \sum_{i=1}^n \ln p_i - \ln (p_i - 1/K)=\sum_{i=1}^n \ln \left(\frac{p_i}{p_i-1/K}\right)\\
 & = \sum_{i=1}^n\ln\left(\frac{1}{1-1/(Kp_i)}\right) \leq -n\ln\left(1-\frac{1}{Ke^{-c_{max}/M}}\right).
\end{align*} 
Since this upper bound holds for every $x\in\{0,1\}^n$ it also applies to the maxima of~\eqref{FORM3}, and by Lemma~\ref{CNVXLEM} it suffices to choose $K$ so that
\begin{align*}
& -n\ln\left(1-\frac{1}{Ke^{-c_{max}/M}}\right) &&<&&  \frac{1}{5M^2} && \leq \min_{L\neq M}\{f(M)-f(L)\} \\
&\Leftrightarrow 
 1-\frac{e^{c_{max}/M}}{K} &&>&& e^{-\frac{1}{5nM^2}}
&&\Leftrightarrow K > \frac{e^{c_{max}/M}}{1-e^{-\frac{1}{5nM^2}}}>\frac{e}{1-e^{-\frac{1}{5nM^2}}}. 
\end{align*} By a Taylor series expansion of the denominator it follows that any  
$K>\frac{5nM^2}{1-1/(10nM^2)}$ 
is sufficiently large. 
\end{proof}

\begin{proposition}
The all-or nothing subset problem~\eqref{FORM3} is NP-hard. 
\end{proposition}
\begin{proof} We prove our claim by providing a reduction of the subset sum problem with positive integer inputs, which is known to be NP-hard. Consider an instance where the goal is to decide whether there exists a subset of $\{c_1,\dots,c_n\}\subset\nat$ that sums to $M\in\nat$. Without losing of generality it is assumed that $c_{\text{max}}\leq M$. For each $i\in[n]$, let $p_i = e^{-\frac{c_i}{M}}$. Let $K\in O(nM^2)$ be an integer satisfying the condition of Lemma~\ref{LEM1} (which also states that it suffices to choose $K=6nM^2$). Set $\hat p_i = \lfloor p_i K\rfloor / K\geq  e^{-\frac{c_i}{M}}-\frac{1}{K}$ for each $i\in[n]$.  
Following Corollary~\ref{COR1} the subset sum problem has a feasible solution if and only if 
\[
\max_{x\in\{0,1\}^n} z(x) = \ln M - 1 > \max\left\{\ln(M-1)-\frac{M-1}{M},\ln(M+1)-\frac{M+1}{M}\right\}.
\] 
By the choice of $K$ and Lemma~\ref{LEM1} it follows that 
\[
\max_{x\in\{0,1\}^n} z(x) - \max_{\{0,1\}^n}\hat z(x) 
< \ln M - \max\left\{\ln(M-1)-\frac{1}{M}, \ln(M+1)+\frac{1}{M}\right\}\] 
Then it follows there exists an $x\in\{0,1\}^n$ such that $\sum_{i=1}^nc_ix_i=M$ if an only if \[
\max_{x\in\{0,1\}^n}\hat z(x)>\max\left\{\ln(M-1)-\frac{M-1}{M},\ln(M+1)-\frac{M+1}{M}\right\}. 
\] Since $K\in O(nM^2)$ it follows that the reduction of subset sum is polynomial in $n$, $\ln M$ and $\ln c_{max}$.
\end{proof}

\section{Approximation of Maximum Expected Value All-or-Nothing Subset}

We now develop an FPTAS for our nonlinear unconstrained problem~\eqref{FORM3}. To this end we first consider a pseudo-polynomial time algorithm. This analysis is similar to that of a related constrained linear problem, namely the knapsack problem;  see~\cite{Ibarra75,Gens80,vazirani01}. A fundamental difference is that~\eqref{FORM3}  unconstrained. 

\subsection{A Pseudo-polynomial Dynamic Program}
For $i\in [n]$ let $P(i,C)$ denote the maximum probability of a subset of $[i]$ with a profit of exactly $C$. 
Consider the dynamic program (DP) given by the equations
\begin{align}
\label{DPEQNS}
P(i,C)=\begin{cases}
\max\{P(i-1,C),  p_i \cdot P(i-1,C-c_i)\} & i\geq 2, c_i < C\\
P(i-1,C) & i\geq 2, c_i\geq C\\
p_1 & i=1 \text { and } c_{1}=C\\
1 & C=0\\
0 & \text{ otherwise.} 
\end{cases}
\end{align}

Let $\bar C$ denote an upper bound on the sum of profits of an item set that is optimal for~\eqref{FORM3}. A straightforward upper bound is $\bar C = \sum_{i=1}^n c_i$. 

Then, the problem of determining $x\in\{0,1\}^n$ that maximizes 
~\eqref{FORM3} is solved by determining 
\begin{equation}
\label{DPSOL}
\max_C\set{C\cdot P(n,C)}{C = \min_{i\in[n]}\{c_i\},\min_{i\in[n]}\{c_i\}+1,\ldots, \bar C}.
\end{equation}
The total running time of this algorithm that determines an optimum of~\eqref{FORM3} through~\eqref{DPSOL} is $O(n\bar C)$. In the following let $x^*\in\{0,1\}^n$ be an optimal solution for~\eqref{FORM3} with support $S^*=\set{i\in[n]}{x^*_i=1}$, and let $C^*=\sum_{i\in S^*}c_i$ denote the corresponding maximizer of~\eqref{DPSOL}.  

\noindent The next lemma establishes a lower bound on the probabilities of items that are included in an optimal solution.

\begin{lemma}
	\label{LEMHALF} Suppose $S^*$ is (the support of a solution that is) optimal for~\eqref{FORM3} with $\card{S^*}\geq 2$, and $l\in\argmin_{i\in S^*}\{p_i\}$.
If $p_l<\frac{1}{2}$ then $\prod_{i\in S^*\setminus\{l\}}p_i\geq\frac{1}{2}$. 
\end{lemma}
\begin{proof}
	Assume for the sake of deriving a contradiction that there exists an $l\in S^*$ with $p_l  <  \frac{1}{2}$ and $\prod_{i\in S^*\setminus \{l\}}p_i < \frac{1}{2}$. Let $X=\prod_{i\in  S^*\setminus\{l\}}p_i\sum_{i\in \hat S^*\setminus\{l\}}c_i$. Then, \[
	\prod_{i\in S^*}p_i\sum_{i\in S^*}c_i=p_l X + p_lc_l\prod_{i\in S^*\setminus \{l\}}p_i < \max\{X,p_lc_l\},
	\] thereby establishing a contradiction with the optimality of $S^*$.
\end{proof}

\noindent In particular Lemma~\ref{LEMHALF} implies the following corollary.

\begin{corollary}\label{ONEHALFCOR}
Suppose $S^*$ is (the support of a solution that is) optimal for~\eqref{FORM3}. Then \[\card{\set{i\in S^*}{p_i < \frac{1}{2}}}\leq 1.
\]
\end{corollary}
\noindent The result of this corollary is instrumental for developing an FPTAS that is the subject of the next section.

\subsection{A Fully Polynomial Time Approximation Scheme}

In order to approximately solve DP~\eqref{DPSOL} and with a polynomial run-time complexity bound we consider scaling down (and rounding) the profit coefficients whose magnitude determines the running time of~\eqref{DPSOL}. In particular consider scaling the profit coefficients using some factor $\kappa > 0$.  Accordingly, for each $i\in[n]$, $\hat c_i = \lfloor\frac{c_i}{\kappa}\rfloor$ is the scaled profit coefficient. In the following let $N_{1/2}=\set{i\in[n]}{p_i\geq \frac{1}{2}}$.  Further, for convenience assume that $[n]\setminus N_{1/2}=[h]$ for some $h\in [n]\cup \{0\}$ ($h=0$ when $[n]\setminus N_{1/2}=\emptyset$). Accordingly, $N_{1/2}=\{h+1,\ldots,n\}$. For $i\in [n]$ let $\hat P(i,C)$ denote the DP equations~\eqref{DPEQNS} with the $\hat c_i$ values in place of the $c_i$'s. Also let $\hat c_{n+1}=0$ and $p_{n+1}=1$. 
Then, for $i\in [n]$ and $j>i$ the scaled DP problem is defined as
\begin{equation}
\label{SCALEDDPSOL}
\hat z(i,j)=\max_C\set{(C+\hat c_j)\cdot \hat P(i,C)\cdot p_j}{C = \min_{k\in[i]}\{\hat c_k\},\min_{k\in[i]}\{\hat c_k\}+1, \ldots, \bar C(i)}, 
\end{equation}
where $\bar C(i) = \sum_{k=1}^i \hat c_k$. 

Note that $\hat z(i,j)$ is an optimal objective value of~\eqref{FORM3} with $c$ replaced by $\hat c$ and the additional constraints (fixing the decision variable values) for $k\in \{i+1,\ldots,n\}$
\[
x_k = 
\begin{cases}
 0 & k\in \{i+1,\ldots,n\}\setminus \{j\}\\
 1 & k = j.	
\end{cases}.
\]
Let $\hat C\equiv \hat z(n,n+1)=max_C\set{C\cdot\hat P(n,c)}{C=1,\ldots,\bar C(n)}$ and let $\hat S$ be the corresponding support of $x$ that maximizes~\eqref{FORM3} with $c$ replaced by $\hat c$ (for which $\sum_{i\in\hat S}\hat c_i=\hat C$). 
Following Corollary~\ref{ONEHALFCOR}, it can observed that it suffices to evaluate $\hat z(i,j)$ with $h=i < j = h+1,\ldots,n+1$ to determine $\hat z(n,n+1)$ and $\hat C \in\argmax_C\set{C\cdot\hat P(n,c)}{C=1,\ldots,\bar C(n)}$.

The following lemma establishes an upper bound on $\kappa$ that is sufficient to bound the relative error to within a given $\epsilon>0$.
\begin{lemma}\label{FPTASLEM}
For a given $\epsilon > 0$, and all $\kappa \leq \frac{\epsilon \max_{i\in S^*}p_ic_i}{n}$, 
\begin{equation}
\label{APPROXINEQ}
  \kappa\cdot \hat z(n,n+1) 
  \geq (1-\epsilon) \cdot C^*\cdot P(n,C^*),
\end{equation} where $C^*$ is a maximizer of~\eqref{DPSOL}.
\end{lemma}
\begin{proof}
First note that 
\begin{equation*}
 \sum_{i\in S^*}c_i-\kappa\sum_{i\in S^*}\hat c_i\leq n\kappa.
\end{equation*}
Then, it follows that 
\begin{align*}
  \kappa\cdot \hat C \cdot \hat P(n,\hat C) =\kappa\sum_{i\in \hat S}\hat c_i\prod_{j\in \hat S}p_j 
& \geq \left(1-\frac{n\kappa}{\sum_{i\in  S^*}c_i}\right)\sum_{i\in S^*}c_i\prod_{j\in  S^*}p_j,
\end{align*} where the last inequality also followed from the optimality of $\hat S$ with the scaled profit $\hat c_i$ values. Then, given an $\epsilon>0$,~\eqref{APPROXINEQ} implies that $\kappa$ must satisfy
\begin{equation*}
\frac{n\kappa}{\sum_{i\in S^*}c_i}\leq \epsilon\Leftrightarrow \kappa\leq \frac{\epsilon\sum_{i\in S^*}c_i}{n},
\end{equation*}
and so it suffices to choose
\[
\kappa\leq \frac{\epsilon \max_{i\in[n]}\{p_ic_i\}}{n}\leq \frac{\epsilon\sum_{i\in S^*}c_i\prod_{j\in  S^*}p_j}{n} \leq \frac{\epsilon\sum_{i\in  S^*}c_i}{n}.
\]
\end{proof}

 Algorithm~\ref{FPTAS} is now considered as an approximation scheme for~\eqref{DPSOL} and (the equivalent)~\eqref{FORM3}.

\begin{algorithm}
\caption{\label{FPTAS}}
\begin{algorithmic}[1]
	\REQUIRE{$\epsilon,c,p$}
	\STATE $\kappa\leftarrow \frac{\epsilon \max_{i\in N_{1/2}}\{p_ic_i\}}{n}$
	\FOR{$j=h+1,\ldots,n+1$}\label{BEGINLOOP}
\STATE $z\leftarrow \max\{\kappa\cdot \hat z(h,j), c_jp_j \} $\label{SINGLEITEMSTEP}
\IF{$z_{\text{max}} < z$}
\STATE $z_{\text{max}} \leftarrow z$
\ENDIF
\ENDFOR\label{ENDLOOP}
\ENSURE{$z_{\text{max}}$}
\end{algorithmic}
\end{algorithm}

\noindent The following proposition establishes that Algorithm~\ref{FPTAS} is an FPTAS for~\eqref{FORM3}. 
\begin{proposition}
Algorithm~\ref{FPTAS} is an FPTAS for~\eqref{FORM3}.
\end{proposition}
\begin{proof}
The following cases need to be considered.
\paragraph{Case $\card{S^*}=1$:} It is straightforward that Algorithm~\ref{FPTAS} outputs an optimal solution determined in step~\ref{SINGLEITEMSTEP}.
\paragraph{Case $\card{ S^*}\geq 2$:} It follows from Corollary~\ref{ONEHALFCOR} that if $\card{ S^*}\geq 2$ then $\card{S^*\setminus N_{1/2}}\leq 1$. Then, consider the following collectively exhaustive subcases:  
\subparagraph{Case $S^*\setminus N_{1/2}=\emptyset$:} 
For each given $\epsilon>0$, $\kappa$ satisfies the supposition of Lemma~\ref{FPTASLEM}. So, following Lemma~\ref{FPTASLEM} with $\bar C=\bar C(h)=\sum_{i\in N_{1/2}}c_i\geq \sum_{i\in S^*} c_i$, 
 \[
 \kappa\cdot \hat z(h,n+1)=\kappa\cdot \hat C\cdot \hat P(h,\bar C)\geq (1-\epsilon) \cdot C^*\cdot P(h,C^*)=P(n,C^*).
 \] 
 
\subparagraph{Case $\card{S^*\setminus N_{1/2}}=1$: } Then 
for each $\epsilon > 0$, the choice of $\kappa$ by Lemma~\ref{FPTASLEM} satisfies for some $j\in [n]\setminus N_{1/2}=\{h+1,\ldots,n\}$ 
\[
\kappa \cdot \hat z(h,j)=\hat z(n,n+1)\geq (1-\epsilon) \cdot C^*\cdot P(n,C^*),
\] and the algorithm must determine $j$ since it enumerates all elements of $[n]\setminus N_{1/2}$ in the main loop (in lines~\ref{BEGINLOOP}-\ref{ENDLOOP}).

The complexity of the algorithm is determined by at most $\card{[n]\setminus N_{1/2}}\leq n$ invocations of~\eqref{SCALEDDPSOL}. Hence, it is \[
O(n^2\bar C)\subseteq O\left(n^2\sum_{i\in N_{\frac{1}{2}}} c_i/\kappa \right)\subseteq O\left(\frac{n^4}{\epsilon}\right).\qedhere
\] 
\end{proof}

\section{Conclusion}

We have established the NP-hardness of all-or-nothing maximum expected value subset. It also implies the hardness of constrained all-or-nothing subset problems in different graph settings. In particular one may consider an all-or-nothing maximum expected value matching, a similar problem, with activities and feasible subsets corresponding to edges and matchings in a graph, respectively. In ongoing work we develop an approximation scheme for this problem.

\section*{Acknowledgement}

Noam Goldberg thanks Naoki Katoh for discussing~\cite{Katoh92} and referring him to~\cite{Matsui96} and~\cite{Hammer02}, John Dickerson for a discussion of kidney exchange and referring to~\cite{Dickerson13}, and also  Martin Milani\v c for comments. 

\bibliography{interdict}

\end{document}